\documentclass[twoside,a4paper]{article}

\usepackage{amsmath,graphicx,amssymb,fancyhdr,amsthm,enumerate,textcomp} 
\usepackage[usenames]{color}

\usepackage{hyperref}
\usepackage{float}
\usepackage{subcaption}
\usepackage[english]{babel}
\usepackage[numbers]{natbib}
\usepackage[utf8]{inputenc}

\usepackage[pagewise,displaymath, mathlines]{lineno}

\newtheorem{thm}{Theorem}[section] 
\newtheorem{cor}[thm]{Corollary} 
\newtheorem{lem}[thm]{Lemma} 
\newtheorem{prop}[thm]{Proposition}
\newtheorem{con}[thm]{Conjecture} 
\theoremstyle{definition} 
\newtheorem{defn}[thm]{Definition}
  
\theoremstyle{remark}  
  
\def\beq{\begin{eqnarray}}  
\def\eeq{\end{eqnarray}}  
\def\bsp{\begin{split}}  
\def\esp{\end{split}}

\def\d{\mathrm{d}}

\newcommand{\mbold}[1]{\mbox{\boldmath{\ensuremath{#1}}}}

\def \bl {\mbox{{\mbold\ell}}}
\def \bn {\mbox{{\bf n}}}
\def \bm {\mbox{{\bf m}}}

\def \bW {\mbox{{\mbold W}}}
\def \bg {\mbox{{\mbold g}}}   
\def \bh {\mbox{{\mbold h}}}   

\begin{document}   
   
\title{\Large\textbf{Locally Homogeneous Kundt Triples and CSI Metrics}}  
\author{{\large\textbf{Sigbj\o rn Hervik and David McNutt }    }
 \vspace{0.3cm} \\     
Faculty of Science and Technology,\\     
 University of Stavanger,\\  N-4036 Stavanger, Norway         
\vspace{0.3cm} \\      
\texttt{sigbjorn.hervik@uis.no, david.d.mcnutt@uis.no} }     
\date{\today}     
\maketitle   
\pagestyle{fancy}   
\fancyhead{} 
\fancyhead[EC]{Hervik and McNutt}   
\fancyhead[EL,OR]{\thepage}   
\fancyhead[OC]{Locally Homogeneous Kundt Triples}   
\fancyfoot{} 

\begin{abstract}

A pseudo-Riemannian manifold is called CSI if all scalar polynomial invariants constructed from the curvature tensor and its covariant derivatives are constant. In the Lorentzian case, the CSI spacetimes have been studied extensively due to their application to gravity theories. It is conjectured that a CSI spacetime is either locally homogeneous or belongs to the subclass of degenerate Kundt metrics. Independent of this conjecture, any CSI spacetime can be related to a particular locally homogeneous degenerate Kundt metric sharing the same scalar polynomial curvature invariants. In this paper we will invariantly classify the entire subclass of locally homogeneous CSI Kundt spacetimes which are of alignment type {\bf D} to all orders and show that any other CSI Kundt metric can be constructed from them. 
\end{abstract}

\section{Introduction} 
In the last decade there has been a program of classifying spacetimes where the elements in the set of all scalar polynomial curvature invariants (SPIs): 
$$ \mathcal{I} = \{ R, R_{abcd} R^{abcd}, \ldots, R_{abcd;e} R^{abcd;e}, \ldots \}, $$ are all constant, such spacetimes are called {\it CSI spacetimes}. For Riemannian spaces, the set $\mathcal{I}$ locally characterizes the manifold completely, and so any Riemannian space with the CSI property must be locally homogeneous (l.h.) \cite{TP}. Clearly, any l.h. space is CSI irrespective of the signature of the metric. For the pseudo-Riemannian spaces there exists classes of manifolds which cannot be uniquely characterized locally by their SPIs which are CSI but not l.h. \cite{Hervik:2010gz, Hervik:2012zz,HHY}. In the particular case of Lorentzian signature there are many examples of CSI metrics which are not l.h. \cite{CSI4a, CSI4b, CSI4c, CSI4d}. 

The CSI spacetimes are of great importance to general relativity (GR) with many applications to alternative gravity theories as well. Important examples in GR can be found within the subclass of CSI spacetimes where all SPIs vanish, the so-called {\it VSI spacetimes}, such as the PP-wave spacetimes and Kundt-wave spacetimes which model massless radiation propagating at the speed of light in a flat universe \cite{Kramer, McNutt:2012}. Both of these solutions and more general VSI spacetimes have been shown to exhibit chaotic geodesic motion \cite{Podolsky:2007,Sakalli:2007}. As another example, exact gravitational waves in anti-de Sitter space can be represented by the Siklos spacetimes which are CSI spacetimes \cite{siklos1981,Podolsky:1997}. These spacetimes have also been shown to contain Ricci solitons \cite{Calvaruso:2019} and admit solutions with non-generic supersymmetries \cite{Baleanu:2001}.   
 
Both the higher-dimensional  VSI and CSI spacetimes are known to  to admit solutions to supergravity theories \cite{CFH0,CFH}. Furthermore, the PP-wave spacetimes have been proven to be solutions to the vacuum equations of all gravitational theories with a Lagrangian constructed from SPIs \cite{horowitz1990spacetime}. Any spacetime with this property is known as a {\it universal spacetime} and this property has been shown to hold for more general VSI spacetimes \cite{coley2008metrics}. Universal spacetimes have been further investigated and it is now known that they are necessarily CSI \cite{coleyhervik2011, Hervik:2013, Hervik:2015, hervik2015, Hervik:2017a, Hervik:2017b, Hervik:2018}.

The condition that a spacetime with metric, ${\bf g}$, is not uniquely characterized locally by its SPIs implies that there exists a smooth (one-parameter) deformation of the metric, ${\bf \tilde{g}}_{\tau}$, with ${\bf \tilde{g}}_{0} = {\bf g}$ and ${\bf \tilde{g}}_{\tau}$, $\tau >0$ not diffeomorphic to ${\bf g}$ yielding the same set $\mathcal{I}$, such a space is called {\it $\mathcal{I}$-degenerate} \cite{CHP2010, HHY}. This definition is often difficult to implement, instead  $\mathcal{I}$-degeneracy can be stated in terms of the structure of the curvature tensor and its covariant derivatives in relation to the action of a boost on a chosen null coframe $\{ \bn, \bl, \bm^i\}$, $\bl' = \lambda \bl,~~ {\bf n'} = \lambda^{-1} \bn,$. For an arbitrary tensor, ${\bf T}$, of rank $n$ the components transform as
\beq T'_{a_1 a_2...a_n} = \lambda^{b_{a_1 a_2 ... a_n}} T_{a_1 a_2 ... a_n},~~ b_{a_1 a_2...a_n} = \sum_{i=1}^n(\delta_{a_i 0} - \delta_{a_i 1}),  \eeq
\noindent where $\delta_{ab}$ denotes the Kronecker delta symbol. The quantity, $b_{a_1 a_2 ... a_n}$, is called the {\it boost weight} (b.w) of the frame component $T_{a_1 a_2 ... a_p}$. Any tensor can be decomposed in terms of the b.w. of its components and hence be classified by identifying null directions relative to which the components of a given tensor have a particular b.w. configuration. This is called the {\it alignment classification}, in 4D it reproduces the Petrov and Segre classifications, while in higher dimensions it provides a coarser classification of tensors \cite{classa,classb,classc, OrtaggioPravdaPravdova:2013}. 

We will define the maximum b.w. of a tensor, ${\bf T}$, for a null direction $\bl$ as the boost order, and denote it as $\mathcal{B}_{{\bf T}}(\bl)$. The Weyl tensor and any rank two tensor, {\bf T}, can be broadly classified into five {\it alignment types}: {\bf G},{\bf I}, {\bf II}, {\bf III}, and {\bf N} if there exists an $\bl$ such that $\mathcal{B}_{{\bf T}} (\bl) = 2, 1, 0,-1,-2$, respectively, and we will say $\bl$ is ${\bf T}$-aligned. If {\bf T} vanishes, then it belongs to alignment type {\bf O}. For higher rank tensors, like the covariant derivatives of the curvature tensor, it is possible that $|\mathcal{B}_{{\bf T}} (\bl)|$ may be greater than two but the alignment types are still applicable. The definition of an $\mathcal{I}$-degenerate spacetime can now be given as a spacetime that admits a null frame such that all of the positive b.w. terms of the curvature tensor and its covariant derivatives are zero in this common frame, that is they are all of alignment type {\bf II} \cite{Hervik2011}.

Any Lorentzian manifold or spacetime admitting a geodesic, shear-free, twist-free, non-expanding null congruence, $\bl$, belongs to the class of Kundt spacetimes \cite{Kramer}. Within the set of $\mathcal{I}$-degenerate spacetimes, {\it the Kundt spacetimes} play an important role. The Kundt spacetimes for which the curvature tensors and its covariant derivatives are of alignment type {\bf II} are known as the {\it degenerate Kundt spacetimes}. In the three-dimensional (3D) and four-dimensional (4D) cases, all $\mathcal{I}$-degenerate spacetimes are contained in the degenerate Kundt spacetimes \cite{CHPP2009}. It is conjectured that any $D$-dimensional $\mathcal{I}$-degenerate spacetime is a degenerate Kundt spacetime \cite{CHP2010}.  Due to the above result, in 3D and 4D, all $CSI$ spacetimes are either l.h. or they belong to the degenerate Kundt class \cite{CSI4c, CSI4b}. In arbitrary dimension it has been proven that VSI spacetimes belong to the degenerate Kundt class \cite{Higher}. Indeed, all known CSI spacetimes which are not l.h. belong to the Kundt class of metrics and there is strong evidence for the following:  
\begin{con}
A Lorentizan CSI space is either: 
\begin{enumerate}
\item{} A l.h. space; or
\item{} A degenerate Kundt metric. 
\end{enumerate}
\end{con}

Through a sequence of papers several results have supported this conjecture. For example, the alignment theorem implies that those CSI spaces which are not l.h., and hence $\mathcal{I}$-degenerate, must be of type II to all orders \cite{Hervik2011}. Furthermore, a frame can be found so that all boost weight 0 components of the curvature tensors are constants (the so-called CSI$_F$ property) \cite{CSI4b}. The subset of $CSI$ spacetimes belonging to the Kundt class are called {\it Kundt-CSI}. For Kundt-$CSI$ metrics, the metric functions must satisfy particular differential equations \cite{CSI4c, CSI4b, CFH}; these conditions will be reviewed in the next section.

While not all CSI spacetimes are l.h., the l.h. spacetimes will be shown to play a vital role: \emph{for any CSI space there exists a l.h space with identical invariants}. For the Kundt CSI class this implies that there is a l.h. Kundt metric having identical invariants to the non-homogeneous Kundt metric. It is these metrics we will investigate here, in fact, we will classify such metrics using the so-called Kundt triples. 

We note that a similar result cannot hold, in general, for arbitrary signatures as the following simple example shows \cite{HHY}: 
\[ g=2\d x(\d y+v\d x)+2\d u(\d v+y^4\d u).\]
This is a four dimensional neutral signature metric and it is CSI. However, it is not l.h. since  it only possesses only 3 Killing vector fields. Furthermore, there is no l.h. metric which shares the same SPIs as this metric \cite{HHY}.  

The paper is organized as follows. In section \ref{sec:CSIKundt}, we review the general form of degenerate Kundt metrics and we discuss the $CSI$ conditions for these spacetimes. In subsection \ref{subsec:KundtTriple} we introduce the  specialized form of the Kundt$^\infty$ metrics and triples and show that any Kundt-CSI metric is related to a l.h. Kundt$^\infty$ triple. In section \ref{sec:CSIcond}, we employ the CSI conditions to determine an eigenvector problem. By considering the dimension of the subspace spanned be these eigenvectors in subsections \ref{Subsec: MinSym}, \ref{SubSec: MaxSym} and \ref{Subsec: IntSym}, we establish the form of all of the Kundt$^\infty$ CSI triples. Finally in section \ref{sec:Conclusion}, we review the results presented in the paper and discuss possible avenues for future work.

\section{Kundt-CSI spacetimes} \label{sec:CSIKundt}

In arbitrary dimension, the metric for a Kundt spacetime can be written as \cite{CSI4d,Podolsky:2014}: 
\beq 
\bg_K=2\d u\left[\d v +\tilde{H}(v,u,x^k)\d u+\tilde{W}_{i}(v,u,x^k)\d x^i\right]+h_{ij}(u,x^k)\d x^i\d x^j,\label{kundt} 
\eeq 
\noindent where $h_{ij}$, or without indices, ${\bf h}$,  is a Riemannian metric for the $n-2$ transverse space, which will be called the transverse metric. In a degenerate Kundt spacetime, there exists a kinematic null frame such that all of the positive boost weight (b.w.) terms of the Riemann tensor and all of its covariant derivatives $\nabla^{(k)} (Riem)$ are zero (in this common frame) \cite{CHPP2009}. That is, relative to the alignment classification it is of type {\bf II} to all orders. For a degenerate Kundt spacetime this property restricts the metric functions $H$ and ${\bf W}$ to be of the form: 
\beq \begin{aligned} \tilde{H}(v,u,x^k)&=v^2H^{(2)}(u,x^k)+vH^{(1)}(u,x^k)+H^{(0)}(u,x^k), \\  \tilde{W}_i(v,u,x^k)&=vW^{(1)}_i(u,x^k)+W^{(0)}_i(u,x^k). \end{aligned}
\label{HWdegen}\eeq 

For a spacetime to be CSI, corollary 3.3 in \cite{Hervik2011} there exists a frame where all of the curvature tensor and its covariant derivatives have constant components or they are all of type {\bf II} with all of their b.w. zero components constant. We are interested in the latter case as this contains all $\mathcal{I}$-degenerate CSI spacetimes. Choosing a null frame adapted to the geodesic, shear-free, twist-free, non-expanding null congruence $\bl$ it follows that a degenerate Kundt spacetime is CSI if the b.w. zero components of the curvature tensor and its covariant derivatives are constant, and this condition will place further restrictions on the metric functions \eqref{HWdegen} and the transverse Riemannian metric. From theorem 4.1 in \cite{CSI4d} the constancy of the transverse space's curvature tensor and its covariant derivatives requires that ${\bf h}$ must be a l.h. metric and theorem 4.1 in \cite{CSI4b} provides a coordinate transformation to always remove dependence on the $u$ coordinate so that, $h_{ij,u} = 0$. In what follows we will always assume there is no $u$-dependence in the transverse metric ${\bf h}$.

There are additional conditions that must be imposed on the metric functions that arise from the b.w. 0 components of the Riemann curvature tensor to ensure that any SPI formed by the curvature tensor are constant:
\beq \begin{aligned} \tilde{R}_{ 1 2 1 2}&= -2H^{(2)} +\frac 14\left(W^{(1)}_{{i}}\right)\left(W^{i(1)}\right),  \\
\tilde{R}_{ 1 2 i j}&= W^{(1)}_{[ i; j]},  \\
\tilde{R}_{ 1 i 2 j}&= \frac 12\left[W^{(1)}_{ i; j}-\frac 12 \left(W^{(1)}_{  i}\right) \left(W^{(1)}_{ j}\right)\right],  \\
\tilde{R}_{ i j kl}&=R_{ i j k l}. \end{aligned} \label{RiemBW0} \eeq 
\noindent where $R_{ijkl}$ denotes the curvature tensor of the transverse space. Requiring that the above components are constant is known as the CSI$_0$ condition and yields differential equations for the metric functions $W_i$ and $H$ for a degenerate Kundt metric:
\beq \begin{aligned}  & \sigma= H^{(2)} - \frac{1}{4} W^{(1)}_{i} W^{i(1)}, \\ 
& a_{ij} = W^{(1)}_{[i;j]}, \\
& s_{ij} = W^{(1)}_{(i;j)} - \frac12 W^{(1)}_{i} W^{(1)}_{j},  \end{aligned} \label{IPDeqn0} \eeq
\noindent where $\sigma, a_{ij}$ and $s_{ij}$ are constants.
Similarly, we must impose conditions on the components of the covariant derivative of the curvature tensor, known as the CSI$_1$ condition: 
\beq \begin{aligned} &  \tilde{R}_{121i;2} = \alpha_i =  \sigma W^{(1)}_{i} - \frac12 (s_{ij} + a_{ij}) W^{(1) j}, \\
& \tilde{R}_{1ijk;2} = \beta_{ijk}  = W^{(1) l } R_{lijk} - W^{(1)}_{i} a_{jk} + (s_{i[j}+a_{i[j})W^{(1)}_{k]}. \end{aligned} \label{IPDeqn1} \eeq

\noindent Using the CSI$_0$ conditions in conjunction with the Ricci identities, $\beta_{ijk}$ can be rewritten as
\beq \beta_{ijk} = (s_{ij} + a_{ij})_{;k} - (s_{ik} + a_{ik})_{;j}. \label{eq:beta0} \eeq

\subsection{Kundt$^\infty$ triples and Kundt$^\infty$ metrics.} \label{subsec:KundtTriple}
In what follows, we will consider an $(n+2)$-dimensional manifold, $\tilde{M}$, with a Lorentzian metric and an $n$-dimensional manifold, $M$, equipped with a Riemannian metric acting as the transverse space for the Lorentzian metric. 

\begin{defn}
A Kundt$^{\infty}$ metric is of the form: 
\beq
\bg = 2\d u\left(\d v+v^2H\d u+v\bW\right)+\bh,
\eeq
where $\bh$ is a Riemannian metric on some manifold $M$, $\bW$ is a one-form on $M$; i.e, $\bW\in \Omega^1(M)$, and $H$ is a function on $M$, $H\in \Omega^0(M)$. 
\end{defn}
\noindent Due to theorem 4.1 in \cite{CSI4b}, ${\bf h}$ will always be independent of the coordinate $u$. The Kundt$^{\infty}$ metrics appear in the study of the near horizon geometry of a degenerate Killing horizon \cite{MPS2018} and more generally, the near horizon geometry of non-expanding null surfaces \cite{Lewandowski2018}. Due to the generic form of the Kundt${^\infty}$ metrics we can instead study an equivalent definition:
\begin{defn}
A Kundt$^\infty$ triple on $M$, is a triple $(H,\bW,\bh)$ where $\bh$ is a Riemannian metric on $M$, $\bW\in \Omega^1(M)$, and $H\in \Omega^0(M)$.
\end{defn}
\noindent Following from this, we can define a l.h. Kundt$^\infty$ triple.
\begin{defn}
A locally homogeneous Kundt$^\infty$ triple on $M$, is a Kundt$^\infty$ triple where the corresponding Kundt$^{\infty}$ metric is locally homogeneous. 
\end{defn}
For any l.h. Kundt$^\infty$ triple the function $H$ is of the form: 
\beq H=\frac{\sigma}{2}+\frac 18 ||\bW||^2, \label{Heqn} \eeq
where $\sigma$ is a constant, and $||\bW||$ is the norm of $\bW$ as induced by the metric $\bh$. Therefore, the l.h. Kundt$^\infty$ triples are classified by $\bW$ and $\bh$, and for a given dimension we wish to classify the l.h. Kundt$^\infty$ triples or equivalently the l.h. Kundt$^\infty$ spaces.

Some light can be shed on the structure of these metrics. By slightly rewriting the metric, assuming $v\neq 0$: 
\beq
\bg = 2v\d u\left(\frac{\d v}{v}+Hv\d u+\bW\right)+\bh,
\eeq
one sees that the one-forms $v\d u$ and $\d v/v$ are left-invariant one-forms on the 2-dimensional solvable group. There are two obvious choices for the Kundt triples that will yield a (locally) homogeneous space: 
\begin{itemize}
\item{} A l.h. metric, $\bh$, on $M$ with $\bW=0$ and $H$ constant. 
\item{} A l.h. metric, $\bh$, on $M$ with a left-invariant one-form on $M$, $\bW$, and $H$ of the form \eqref{Heqn}.
\end{itemize}
\noindent These are the trivial cases and we will aim to determine the existence of non-trivial l.h. Kundt$^\infty$ triples and classify them. 

To show how the Kundt$^\infty$ metrics are important in the classification of Lorentzian CSI metrics, we will consider the diffeomorphism, $\phi_t$, generated by an appropriate boost in the null plane spanned by $\bl$ and $\bn$ with respect to a point $\tilde{p} \in \tilde{M}$ for a Kundt-CSI metric, $\bg_K$. We will show the resulting metric in the limit $\lim_{t\rightarrow\infty}\phi_t^*\bg_{K}= \bg_0$ is a l.h. Kundt$^\infty$ metric   \cite{HHY}. Explicitly, if in coordinates $(v,u,x^i)\big{|}_p=(0,u_0,x^i_0)$, then the limiting Kundt$^\infty$ triple is given  by 
\[ (H,\bW,\bh)= \left(H^{(2)}(u_0,x^k),~~ W^{(1)}_i(u_0,x^k)\d x^i, ~~h_{ij}(u_0,x^k)\d x^i\d x^j\right).\] 
\noindent The choice of $v=0$ may be chosen without loss of generality since the mapping $v \mapsto v-v_0$ is an $\mathcal{I}$-preserving diffeomorphism for degenerate Kundt metrics \cite{SH2018, MMTA2018, MTA2018}. 

With this limiting process we can prove a helpful result for CSI Kundt$^\infty$ triples.
\begin{prop}\label{prop:Kundt}
Given a degenerate Kundt-CSI metric, there is a corresponding Kundt$^\infty$ triple,  $(H,\bW,\bh)$, with the following properties: 
\begin{enumerate}
\item{} The Kundt$^\infty$ metric is of algebraic type {\bf D}$^k$ (type {\bf D} to all orders). 
\item{} The Kundt$^\infty$ metric is characterized by its invariants.
\item{} The Kundt$^\infty$ metric has identical invariants to the $\mathcal{I}$-degenerate Kundt-CSI metric. 
\item{} $(H,\bW,\bh)$ is a locally homogeneous Kundt$^\infty$ triple.
\end{enumerate}
\end{prop}
\begin{proof}
The diffeomorphism $\phi_t$ is generated by a boost, ${\bf X}$. In the limit the curvature tensors are necessarily invariant under the action of $\phi_t$ and it is thus an isometry of $\bg_0$ \cite{HHY}. Consequently, it is of type {\bf D}$^k$ and a doubly-degenerate Kundt spacetime \cite{CHPP2009}. To show the second property, we note that type {\bf D}$^k$ tensors have closed orbits under the action of the Lorentz group \cite{Hervik2011} and are therefore characterized by their invariants. The third property also follows from the action of $\phi_t$, this is a continuous mapping which leaves any SPI of the original space invariant for all finite $t$, and hence the limit metric has necessarily identical invariants as well \cite{HHY}. The fourth follows from the fact that it is CSI and of type {\bf D} to all orders. Since the SPIs characterize the components of the curvature tensor and its covariant derivatives \cite{CHP2010} there exists a frame where the components of the curvature tensor and its covariant derivatives are all constant. In analogy with the Riemannian space where the condition that the corresponding curvature tensor and its covariant derivatives have constant components imply that the Riemannian metric is l.h. \cite{TP}, any type ${\bf D}^k$ spacetime with constant components must have a l.h. metric \cite{CHP2010}.  
\end{proof}

The limit of the diffeomorphism $\phi_t$ as $t\to \infty$ can be applied to any  $\mathcal{I}$-degenerate CSI spacetimes as well, since the alignment theorem implies that they are necessarily of alignment type {\bf II} to all orders with constant b.w. zero components \cite{Hervik2011}. Following the proof of theorem 2.1 for Lorentzian metrics in \cite{Hervik:2012zz} and the discussion of theorem 3.1 in \cite{HHY}, since the diffeomorphism $\phi_t$ is generated by a boost, in the limit $t \to \infty$ the resulting tensors are well defined at any point, with all negative b.w. components going to zero while the b.w. zero  components are unchanged and constant by assumption. That is, at each point of the manifold, the resulting curvature tensor and its covariant derivatives are all of type {\bf D} with constant components and we can associate a l.h. Kundt$^\infty$ triple to this point. Furthermore, as the original $\mathcal{I}$-degenerate spacetime is CSI, this procedure will give the same l.h Kundt$^\infty$ triple at any point. 

\begin{cor}
For any $\mathcal{I}$-degenerate CSI spacetime, there is a corresponding l.h. Kundt$^\infty$ triple with an identical set of SPIs $\mathcal{I}$.
\end{cor}

Given a degenerate Kundt CSI metric $\bg_{K}$, taking the limit of the diffeomorphism $\phi_t$, with respect to a point $\tilde{p} \in \tilde{M}$ generated by the boost defined in \cite{HHY}  produces a  Kundt$^\infty$ CSI metric, $\bg_0$, 

$$\lim_{t\rightarrow\infty}\phi_t^*\bg_{K}= \bg_0.$$

\noindent This limit destroys any information about the functions $H^{(1)}, H^{(0)}$, and $W^{(0)}$ in the original Kundt CSI metric. A new generic Kundt-CSI metric can be constructed by adding arbitrary functions, $H^{(1)}, H^{(0)}$, and $W^{(0)}$ to the metric functions in \eqref{HWdegen} of a l.h. Kundt$^\infty$ metric. The addition of these functions will contribute to the negative b.w. components of the curvature tensor and its covariant derivatives, with no effect to the b.w. zero components, therefore it will be of alignment type {\bf II} to all orders and hence $\mathcal{I}$-degenerate \cite{Hervik2011}. As this metric has the same $H^{(2)}$, $W^{(1)}_i$ and l.h. transverse metric $\bh$, the b.w. zero components of the curvature tensor and its covariant derivatives will be identical and will produce the same constant SPIs as the original l.h. Kundt metric \cite{CHP2010}. Requiring that the Riemann tensor and its covariant derivatives are of a more special secondary algebraic type with respect to $\bn$, will yield differential equations for the metric functions $H^{(1)}, H^{(0)}$, and $W^{(0)}$.  

The relationship between a non-l.h. Kundt-CSI metric and the corresponding l.h Kundt$^\infty$ triple can be summarised as:

\begin{prop} \label{prop:KundtCSI}
Given a Kundt-CSI metric, $\bg_K$. If the metric is not locally homogeneous then the metric is of the form: 
\beq 
\bg_K=2\d u\left[\d v +\tilde{H}(v,u,x^k)\d u+\tilde{W}_{i}(v,u,x^k)\d x^i\right]+h_{ij}(x^k)\d x^i\d x^j,
\label{KCSI}\eeq
with 
\beq \tilde{H}(v,u,x^k)&=&v^2H^{(2)}(x^k)+vH^{(1)}(u,x^k)+H^{(0)}(u,x^k), \nonumber \\  \tilde{W}_i(v,u,x^k)&=&vW^{(1)}_i(x^k)+W^{(0)}_i(u,x^k)
\label{CSI}\eeq 
 and the triple
\[ (H,\bW,\bh)= \left(H^{(2)}(x^k),~~ W^{(1)}_i(x^k)\d x^i, ~~h_{ij}(x^k)\d x^i\d x^j\right)\] 
is a locally homogeneous Kundt$^{\infty}$ triple. 
\end{prop} 
\begin{proof}
Since the space is not l.h. but CSI it must be of aligned type {\bf II} (or simpler, not {\bf D}) due to corollary 3.3 in \cite{Hervik2011}. As the metric has been assumed to be Kundt it is necessarily degenerate Kundt and of the form \eqref{kundt}. The CSI conditions for the curvature tensor and its covariant derivatives restricted to the transverse space imply that a coordinate transformation can be used to remove all dependence on the coordinate $u$ in $h_{ij}$ in \eqref{kundt} \cite{CSI4b}. 
At every point one can now apply a limit of the boost given in the proof of proposition \ref{prop:Kundt}: $\lim_{t\rightarrow\infty}\phi_t^*\bg_{K}= \bg_0$ to show that the related Kundt triple $(H, {\bf W}, {\bf h})$ is l.h. 

\end{proof}

For the remainder of the paper, when discussing Kundt$^{\infty}$ triples, we will denote the metric function $H^{(2)}$ as $H$ and the one-form ${\bf W}$ components, $W^{(1)}_i$ as $W_i$. 

\section{CSI conditions for Kundt$^\infty$ triples} \label{sec:CSIcond}

We will show that if the CSI$_0$ and CSI$_1$ conditions are satisfied for a Kundt$^\infty$ metric then it will be CSI. From proposition \ref{prop:KundtCSI} this will imply that any Kundt-CSI spacetime is CSI if the CSI$_0$ and CSI$_1$ conditions are fulfilled. The CSI$_0$ and CSI$_1$ conditions in equations \eqref{IPDeqn0} and \eqref{IPDeqn1} can be written in a simpler form for the Kundt$^\infty$ triples \cite{CSI4b}: 

\beq \d \bW =W_{[i;j]}\omega^i\wedge\omega^j&=&\frac{1}2{\sf a}_{ij}{\omega^i}\wedge{\omega^j},\label{eq:a}\\
W_{(i;j)}-\frac 12 W_iW_j&=&{\sf s}_{ij},\label{eq:s}\\
\sigma W_{i}-\frac 12({\sf s}_{ij}+{\sf a}_{ij})W^j&=&\alpha_i,\label{eq:alpha} \\
W^n R_{nijk} - W_i {\sf a}_{jk} + ({\sf s}_{i[j} + {\sf a}_{i[j}) W_{k]} &=&\beta_{ijk}, \label{eq:beta} 
 \eeq
where ${\sf a}_{ij}$, ${\sf s}_{ij}$, $\alpha_i$ and $\beta_{ijk}$ are left-invariant tensors on $\bh$. 

Since $\bh$ must be a l.h. space to ensure that the components of the transverse space's curvature tensor and its covariant derivatives are constant \cite{TP, CSI4d} there exists (locally) a set of $n$ linearly independent Killing vector fields, $\xi_k$, $k=1,...,n$ so that the Lie derivatives of ${\sf a}_{ij},{\sf s}_{ij}$ and $\alpha_i$ with respect to $\xi_k$ all vanish. 

Taking the Lie derivative of equation (\ref{eq:alpha}), yields
\beq 2\sigma v_{i}-({\sf s}_{ij}+{\sf a}_{ij})v^j=0,\label{eq:v}\eeq
where ${\bf v}:=\pounds_{\xi_k}{\bW}$, for some $k$. Consequently, ${\bf v}$  is an eigenvector of $({\sf s}_{ij}+{\sf a}_{ij})$ with eigenvalue $2\sigma$. If ${\bf v}={\bf 0}$, then this is automatically satisfied, but if ${\bf v}\neq {\bf 0}$, then $({\sf s}_{ij}+{\sf a}_{ij})$ must have a non-trivial eigenvector with eigenvalue $2\sigma$.  Let us therefore define the vector space:
\beq
I^1=\left\{\pounds_{\xi_k}\bW~\big{|}~\xi_k\text{ is Killing} \right\}, \label{def:I1}
\eeq
and, iteratively, 
\beq
I^{n+1}=\left\{\pounds_{\xi_k}{\bf v}~\big{|}~\xi_k\text{ is Killing}, {\bf v}\in I^n\right\}.
\eeq
Now, we quickly see that for any ${\bf v}\in I^n$, equation (\ref{eq:v}) must hold. Consider therefore a point, $p\in M$, and let 
\beq {\cal I}_p={\rm span}\{I^1\big{|}_p, I^2 \big{|}_p , ... , I^n\big{|}_p, ...  \}\subset T_pM. \label{Ipset} \eeq
The following analysis now breaks into cases considering the different dimensions of ${\cal I}_p$ with $0\leq \dim{\cal I}_p\leq n$. We will consider the extreme cases first, when $\dim{\cal I}_p=0$ and $\dim{\cal I}_p=n$. 

\subsection{$\dim{\cal I}_p=0$} \label{Subsec: MinSym}
Here, $\pounds_{\xi_k}{\bW}=0$ at $p$, if this is true over a neighbourhood, then $\bW$ must be left-invariant with respect to  the set $\{\xi_k \}$. This further implies that the remaining equations \eqref{eq:a} and \eqref{eq:s} are identically satisfied as well.

\subsection{$\dim{\cal I}_p=n$} \label{SubSec: MaxSym}
This is the maximal dimension and equation (\ref{eq:v}) implies that
\[ ({\sf s}_{ij}+{\sf a}_{ij})=2\sigma\delta_{ij}.\]
Hence the anti-symmetric part must vanish, ${\sf a}_{ij}=0$, and the remaining symmetric piece is diagonal, ${\sf s}_{ij}=2\sigma\delta_{ij}$ or equivalently  ${\sf s}_{ij}= 2 \sigma h_{ij}$ since the frame basis for the transverse space is orthonormal.
\begin{lem} \label{lem:Maxsym}
If $\dim{\cal{I}}_p=n$ for all $p\in M$, then $(M,{\bh})$ is locally maximally symmetric with  Riemann tensor of the form:
\[ {R}_{ijlm}=\sigma\left(h_{il}h_{jm}-h_{im}{h_{jl}}\right).\]
\end{lem}
\begin{proof}
	The equivalent form of $\beta_{ijk} = R_{1ijk;2}$ in \eqref{eq:beta0} along with the conditions that ${\sf a}_{ij}=0$ and  ${\sf s}_{ij}=2\sigma\delta_{ij}$ forces $\beta_{ijk} =0$. Substituting ${\sf s}_{ij}=2\sigma\delta_{ij}$ and ${\sf a}_{ij} = 0$ into the original definition of $\beta_{ijk}=0$ in \eqref{eq:beta}, then by repeated application of the Lie derivative with respect to the Killing vector fields we can determine the components $R_{lijk}$, since there is an associated orthonormal basis $\{ \omega^i \}$ of $\mathcal{I}_p$.
	\end{proof}
To determine ${\bW}$, we must solve the equations (\ref{eq:a}) and (\ref{eq:s}) with  ${\sf a}_{ij}=0$, and ${\sf s}_{ij}=2\sigma h_{ij}$. The first implies $\d\bW=0$, and hence, $\bW$ can be locally written as $\bW=\d \phi$. Then, using the substitution $f=\exp(-\phi/2)$, equation (\ref{eq:s}) gives the following linear equation: 
\beq
f_{;ij}=-\sigma h_{ij}f.
\label{eq:f}
\eeq

Let us consider the $\sigma=0$ (flat) case first as we will see all the other cases can be lifted to the flat case of one dimension higher. 
That is, we will examine $\mathbb{E}^n$, with $\sigma=0$ and employ Cartesian coordinates so that the covariant derivatives are ordinary derivatives and we must solve the equation: 
\[ \partial_i\partial_j f=0.\]
This has the general solution $f=\alpha_ix^i+b$, where $\alpha_i$ is a constant 'vector' and $b$ is a constant. Using the isometries of $\mathbb{E}^n$ consisting of rotations, $SO(n-1)$, and translations, we can rotate so that $\alpha_i=(\alpha_1,0,...,0)$. Then using translations, assuming $\alpha_i\neq 0$, we can set $b=0$. This then gives:
\[ \bW=\d(-2\ln f)=-\frac{2}{x^1} \d x^1.\]
If $\alpha_i=0$ then we get 
\[ \bW=\d (-2\ln f)={\bf 0},\] 
so this is the trivial case. 

For $\sigma\neq 0$, there is no loss of generality to assume $\sigma=\epsilon=\pm 1$ (up to scaling). We can now lift the metric $\bh$ to the flat cone by considering $\bh$ embedded in $\tilde{g}=\epsilon \d r^2+r^2\bh$ (flat space) at $r=1$, and employ the following \cite{deMedeiros}:

\begin{thm} \label{thm:cone}
Let $\tilde{g}=\epsilon \d r^2+r^2\bh$ be the lift of $\bh$ to the flat cone of dimension $(n+1)$. Then the function $\Phi$ is a solution to $\tilde\nabla_a\tilde\nabla_b\Phi=0$ if and only if $\Phi(r,x^i)=rf(x^i)+b$ where $f$ is a solution to equation (\ref{eq:f}).
\end{thm}
\begin{proof}
By introducing orthonormal co-frames $\left\{{\omega}^i\right\}$ on $\bh$ and $\left\{\tilde{\omega}^a\right\}=\left\{\d r,{\omega}^i\right\}$ on $\tilde{g}$, then we find the non-zero rotation coefficients: 
\[ \tilde{\Gamma}^i_{~rj}=\frac 1r \delta^i_{~j}, \qquad \tilde{\Gamma}^i_{~jk}=\frac 1r\Gamma^i_{~jk}. \] 
Thus, for a function $F(r,x^i)$:
\[ \d F=F_{;r}\d r+F_{;i}\tilde{\omega}^i=F_{;r}\d r+\frac 1rF_{;i}{\omega}^i.\]
Consequently: 
\beq
\tilde{\nabla}_r\tilde{\nabla}_rF&=&F_{,rr},\\
\tilde{\nabla}_r\tilde{\nabla}_iF&=& -\frac 1{r^2}F_{,i}+\frac 1rF_{,ir},\\
\tilde{\nabla}_i\tilde{\nabla}_jF&=&\frac 1{r^2}F_{,ji}-\frac{1}{r^2}\Gamma^k_{~ji}F_{,k}+\frac 1r\epsilon \delta_{ij}F_{,r}=\frac{1}{r^2}\nabla_i\nabla_jF+\frac 1r\epsilon \delta_{ij}F_{,r}.\label{ij-eq}
\eeq
So requiring $\tilde\nabla_a\tilde\nabla_b\Phi=0$ we get $\Phi_{,rr}=0$ implying $\Phi=rf(x^i)+b(x^i)$ for some functions $f$ and $b$. Next, the $(ir)$-components give $b_{,i}=0$, consequently, $b$ is a constant. 
Finally, the equations $\tilde\nabla_i\tilde\nabla_j\Phi=0$ and equation (\ref{ij-eq}) give: 
\[ \frac{1}{r}\nabla_i\nabla_jf+\frac 1r\epsilon\delta_{ij}f=0, \] which, in an arbitrary frame can be written as: 
\[ \nabla_i\nabla_jf=-\epsilon h_{ij}f.\] 
The theorem is now proven. 
\end{proof}

The lift enables us to consider the cases $\sigma=+1$ and $\sigma=-1$ when $(\bh,M)$ are (locally) the $n$-sphere $S^n$ and the hyperbolic space ${\mathbb H}^n$, respectively.
 Note also that the isometries of $S^n$ and ${\mathbb H}^n$ are the isometries of $\mathbb{E}^{n+1}$ and $\mathbb{E}^{n,1}$  that leave the sphere and hyperboloid, respectively, invariant. The corresponding groups are $SO(n)$ and $SO(n-1,1)$, respectively, and are the isotropy groups of the origin of the flat cone.
 
 For each case, the general solution of the equation $\tilde\nabla_a\tilde\nabla_b\Phi=0$ can be obtained using Cartesian coordinates $(x^a)$. In particular, 
 \[ \tilde{g}=\epsilon(dx^1)^2+(dx^2)^2+...+(dx^{n+1})^2,\]
 so that the connection is trivial and the covariant derivatives are ordinary derivatives. Hence, 
 \[ \Phi=\alpha_a x^a+b,\]
 where $\alpha_a$ are constants and $b$ is a constant.
 
We can now utilize the isometries of $S^n$ and ${\mathbb{H}}^n$ to simplify the 'vector' $\alpha_a$ so that there is one case for $S^n$,
\beq  \Phi=\alpha_1x^1+b, \nonumber \eeq
\noindent and three cases for $\mathbb{H}^n$:
\begin{enumerate}
\item{} Timelike case: $\Phi=\alpha_1x^1+b$.
\item{} Lightlike case: $\Phi=(x^1+x^2)+b$.
\item{} Spacelike case: $\Phi=\alpha_2x^2+b$.
\end{enumerate}

As we are interested in the one-form ${\bW} = -2 d \ln f$ we may ignore the constant $b$ so that $\Phi = r f(x^i)$ in the original coordinates and we can consider the lifted one-form ${ \hat{\bW}} = -2 d \ln r + {\bW}$. Using this one-form we can prove the following result relating $\mathcal{I}_p$ to $I^1|_p$. 

\begin{cor} \label{cor:rankI1}
For a l.h. Kundt triple with a maximally symmetric Riemannian manifold, $(M, {\bh}$),  if the one-form, ${\bW}$, is not left-invariant then
\beq dim I^1|_p = dim \mathcal{I}_p = n. \nonumber \eeq
\end{cor}

\begin{proof}
When $\sigma=0$, no lift is necessary. Due to the specific form of $\bW$ given in this case, $\bW = - \frac{2}{x^1} dx^1$, the infinitesimal generators of rotations in the $x^1$-$x^i$ plane $i \in [2,n]$ along with translations in the $x^1$ direction generate a collection of $n$ linearly independent one-forms under Lie differentiation. Therefore, from equation \eqref{def:I1} $dim I^1|_p = n$ at any point in $M$.  

In the $\epsilon = 1$ case, relative to the original coordinates where $\Phi = \alpha_1 x^1 +b$, Lie differentiation of ${ \hat{\bW}}= -2d \ln \Phi$ with the infinitesimal generators of rotations in the $x^1$-$x^i$ plane for $i \in [2,n+1]$ along with translations in the $x^1$ direction yield the necessary $n+1$ linearly independent one-forms. Then by fixing $r=1$ we recover $n$ linearly independent one-forms  on $M$.

When $\epsilon = -1$, there is a subgroup of $SO(n-1,1)$ that will generate $I^1|_p$ for each case depending on the form of $\Phi$ in the original coordinates. For the timelike case, we may use the infinitesimal generators of the boosts in the $x^1$-$x^i$ plane for $i \in [2,n+1]$ and translations in the $x^1$ direction. In the null case, we can employ the infinitesimal generators of $Sim(n-2)$ leaving the null direction $x^1+x^2$ fixed and translations in the $x^1$ and $x^2$ directions. Lastly, in the spacelike case, the infinitesimal generators of rotations about the $x^2$ axis, boosts in the $x^1-x^2$ plane and translations in the $x^2$ direction  will be sufficient. Again, fixing $r=1$ yields the needed $n$ one-forms in $I^1|_p$ to span $\mathcal{I}_p$ for any $p \in M$. 
\end{proof}

Interestingly, this trick of lifting the equations to the flat cone implies not only that we can find all solutions to the CSI equations, but the symmetries imply that many cases given in \cite{CSI4b} are the same. In fact, for example, the $S^n$ case, there is only one (up to isometries) Kundt$^{\infty}$ triple. 

Before we give the explicit metrics, let us state a useful observation. Using the function $f$ we can write $\bW=-2\d\ln f$ (so that $\bW$ does not depend on a constant scaling of $f$). By taking the corresponding Kundt$^\infty$ metric, the coordinate change $f^2\tilde{v}=v$,  allows us to rewrite the metric:
\[ \bg=2f^2\d u\left(\d\tilde v+\tilde{v}^2\mathcal{H} du \right)+\bh,\]
where 
\[ \mathcal{H}=\frac 12\left(\sigma f^2+||\d f||^2\right).\]
Taking the covariant derivative:
\[ \nabla_i\mathcal{H}=(\partial^jf)(\sigma h_{ij}f+f_{;ji});\]
hence, $\mathcal{H}$ is constant if equation (\ref{eq:f}) is satisfied, or if $f$ is constant (then $\bW=0$). Setting $\sigma=\epsilon$ and lifting to the flat cone, we can write:
\[ \mathcal{H}=\frac 12\langle\d \Phi,\d\Phi\rangle_{\tilde{g}}=\frac 12(\partial_a\Phi)(\partial^a\Phi).\]
In particular, this means the timelike, lightlike, and spacelike cases of the $\mathbb{H}^n$ metric have $\mathcal{H}<0$, $\mathcal{H}=0$, and $\mathcal{H}>0$ respectively.

By restricting to $r=1$, we can choose a suitable set of coordinates on $S^n$ and ${\mathbb H}^n$. Then the Kundt metrics for each case are:
 \begin{itemize}
\item{} $S^n$: $f=\cos x$
\[ \bg=2\d u\left[\d v+\frac{v^2}2\left(1+\tan^2x\right)\d u+2v\tan x\d x\right]+\d x^2+\sin^2x\d \Omega^2_{n-1}.\]
\item{} ${\mathbb H}^n$:  
\begin{enumerate}
\item{} Timelike case: $f=\cosh x$
\[ \bg=2\d u\left[\d v+\frac{v^2}2\left(-1+\tanh^2x\right)\d u-2v\tanh x\d x\right]+\d x^2+\sinh^2x\d \Omega^2_{n-1}.\]
\item{} 	Lightlike case: $f=e^{-x}$
\[ \bg=2\d u\left[\d v+2v\d x\right]+\d x^2+e^{-2x}\d E^2_{n-1}.\]
\item{} Spacelike case: $f=\sinh x$
\[ \bg=2\d u\left[\d v+\frac{v^2}2\left(-1+\coth^2x\right)\d u-2v\coth x\d x\right]+\d x^2+\cosh^2x\d H^2_{n-1}.\]
\end{enumerate}
 \end{itemize}

\subsection{$ 0 < \dim{\cal I}_p<n$} \label{Subsec: IntSym}

To study the possible choices for the l.h. transverse space of a given CSI Kundt triple, $(H, {\bW}, {\bh})$ with $ \dim{\cal I}_p \in (0,n)$, we will first prove a useful result about the eigenvectors of $2\sigma$. Motivated by corollary \ref{cor:rankI1}, we will examine the dimension of $I^1|_p$ for an arbitrary point $p\in M$.  In what follows the rank of the eigensubspace $I^1|_p$ is $0 < m \leq dim \mathcal{I}_p < n$, while $k$ is the number of linearly independent Killing vector fields acting on the transverse space. We will assume that $dim~ I^1|_p = m$ for all $p \in M$. 

\begin{lem} \label{lem:KVbasis}
Defining $\tilde{v}^I = \pounds_{\xi_I} {\bW}$, $I \in [1,m]$,  a basis for the Killing vector fields can be chosen such that 
\beq \tilde{v}^I = \pounds_{\xi_I} {\bW},~~\text{and}~~ \pounds_{\xi_{\hat{I}}} {\bW} = {\bf 0},~~\hat{I} \in [m+1, k]. \nonumber \eeq
\end{lem}

\begin{proof}
As the rank of the vector space is $m$, we may always write:  
\beq \pounds_{\xi_{\alpha}} {\bW} = C_{J\alpha} \tilde{v}^J,~~\alpha \in [1,k], \nonumber \eeq

\noindent where $C_{IJ} = \delta_{IJ}$. The coefficients must be constant, since ${\bW}$ satisfies 
\beq \d \bW = {\sf a}_{ij} m^i \wedge m^j \nonumber \eeq
\noindent relative to the left-invariant basis, and so 
\beq \pounds_{\xi} \d {\bW} = \d \pounds_{\xi} {\bW} = {\bf 0} \nonumber \eeq
\noindent for any Killing vector field. This implies $\d \tilde{v}^I = {\bf 0}$ and therefore
\beq  \d \pounds_{\xi_{\alpha}} {\bW} = \d (C_{J\alpha}) \tilde{v}^J = {\bf 0}. \nonumber \eeq
Choosing a new basis where $\tilde{\xi}_I = \xi_I$ and $\tilde{\xi}_{\hat{I}} = B_{\hat{I}}^{~\alpha} \xi_{\alpha}$ with the components of the matrix {\bf B} constant, the requirement that the Lie derivative of the remaining vector fields vanishes gives a constraint on ${\bf B}$ and ${\bf C}$:
\beq \pounds_{\tilde{\xi}_{\hat{I}}} {\bW} = B_{\hat{I}}^{~\alpha} C_{J\alpha} \tilde{v}^J =0. \nonumber \eeq
For fixed $\hat{I}$ this represents vector-matrix multiplication where the rank of the $k \times m$ matrix ${\bf C}$ must be $m$ and the dimension $k$, the rank nullity theorem implies the dimension of the null space of ${\bf C}$ is 
\beq Null( {\bf C}) = k-m. \nonumber \eeq
\noindent Thus we can choose the rows of $B_{\hat{I}}^{~\alpha}$ to be the basis vector fields of the null space so that
\beq \pounds_{\tilde{\xi}_{\hat{L}}} {\bW} = {\bf 0}. \nonumber \eeq
\end{proof}

There is one more issue to address: whether the set of Killing vector fields $\xi_{\hat{I}}$ for which $\pounds_{\xi_{\hat{I}}} {\bW} = {\bf 0}$ generates additional eigenvectors through Lie differentiation of $v^I$: 
\begin{lem}
If the rank of $I^1|_p$ is $m$, no new linearly independent eigenvectors appear from $\pounds_{\xi_{\hat{K}}} \tilde{v}^I = \pounds_{\xi_{\hat{K}}} \pounds_{\xi_I} {\bW}$. 
\end{lem}

\begin{proof}

Choosing $ \tilde{v}^I = \pounds_{\xi_I} {\bW},~ I \in[1,m]$ and $\pounds_{\xi_{\hat{K}}} {\bW} = {\bf 0},~\hat{K} \in [m+1, k]$, we will consider $\pounds_{\xi_\alpha } \pounds_{\xi_I} {\bW}$,~ $\alpha \in [1, k]$, using the Jacobi identity:
\beq \pounds_{\xi_\alpha } \pounds_{\xi_I} {\bW} &=& \pounds_{\xi_I } \pounds_{\xi_\alpha } {\bW} + \pounds_{[\xi_\alpha, \xi_I]} {\bW} \nonumber \\
&=& \pounds_{\xi_I } \pounds_{\xi_\alpha } {\bW} + C^\beta_{~\alpha I} \pounds_{\xi_\beta} {\bW} \nonumber \\
&=& \pounds_{\xi_I } \pounds_{\xi_\alpha } {\bW} + C^J_{~\alpha I} \pounds_{\xi_J} {\bW}. \nonumber  \eeq 
\noindent Restricting indices to the interval $[m+1, k]$, i.e., $\alpha = \hat{K}$, since $\pounds_{\xi_{\hat{K}}} {\bW} = {\bf 0}$  this expression becomes: 
\beq  \pounds_{\xi_{\hat{K}} } \pounds_{\xi_I} {\bW} = \sum_{J} C^J_{~\hat{K} I} \tilde{v}^J, \nonumber  \eeq
\noindent 
\end{proof}
The only way to generate additional linearly independent eigenvectors is by repeated application of the Lie derivatives with respect to the set $\{ \xi_I \}$, and this will not impact the set of Killing vector fields where $\pounds_{\xi_{\hat{K}} }  {\bW} = 0$.

\begin{lem} \label{lem:Curv}

Relative to the orthonormal basis adapted to the eigenvectors arising from Lie differentiation of ${\bf {\bW}}$ with respect to the set of Killing vector fields $\{ \xi_{I}\}_{I=1}^{m}$,  the non-zero Riemann tensor components are $$R_{IJKL},~ R_{I \hat{j} K \hat{l}}, \text{ and } R_{\hat{i} \hat{j} \hat{k} \hat{l} }$$ where $I,J,K,L \in [1, m]$ and $ \hat{i}, \hat{j}, \hat{k}, \hat{l} \in [m+1, n]$. The components $R_{IJKL}$ are of the form:
\beq R_{IJKL} &=& \sigma [\delta_{IK} \delta_{JL} -\delta_{IL} \delta_{JK} ]. \nonumber \eeq

\end{lem}

\begin{proof}
Taking the Lie derivative of ${\bf {\bW}}$ with respect to the $m$ Killing vector fields yields a basis of $m$ eigencovectors, $\{\tilde{v}^I\}$, which are related to the orthonormal basis $\{v^I\}$ on the eigenspace $I^1|_p \subset T_pM$ through an invertible matrix $\tilde{v}^I = A^I_{~L} v^L$. This basis can be extended to produce a coframe which diagonalizes the metric, $\{ \tilde{m}^i\} = \{ v^1, \cdots v^m, m^{m+1},\cdots, m^n\}$, the matrix ${\sf s}_{ij} + {\sf a}_{ij}$ takes the form:
\beq {\sf s}_{ij} + {\sf a}_{ij} = \begin{cases} 
      2 \sigma \delta_{ij} & i,j \in [1, m] \\
      0 & i \in [m+1, n],~~j \in [1,m] \\
      {\sf s}_{ij } + {\sf a}_{ij} &  i,j \in [m+1, n]. 
   \end{cases} \nonumber \eeq
Taking the Lie derivative of equation \eqref{eq:beta} with respect to $ \xi_{N}$ and expressing this relative to the new coframe:
\beq 0 &=& (\pounds_{\xi_{N}} {\bW}^n) R_{nijk} - (\pounds_{\xi_{N}} {\bW}_i) {\sf a}_{jk} + ({\sf s}_{i[j} + {\sf a}_{i[j}) (\pounds_{\xi_{I}} {\bW}_{k]}) \nonumber \\
&=& A_{N}^{~L} \delta_L^{~n} R_{nijk} - A_{N}^{~L} \delta_{Li} {\sf a}_{jk} + ({\sf s}_{i[j} + {\sf a}_{i[j} ) A_{|N}^{~~L} \delta_{L|k]} \nonumber \\
&=& A_{N}^{~L} [ R_{Lijk} - \delta_{Li} {\sf a}_{jk} + ({\sf s}_{i[j} + {\sf a}_{i[j} )  \delta_{L|k]}]. \nonumber  \eeq

We can consider six cases for the remaining indices $(ijk)$, $$\{ (IJK), (\hat{i} \hat{j} \hat{k}), (I\hat{j} \hat{k}), (IJ\hat{k}), (\hat{i}JK), (\hat{i}J\hat{k}) \}.$$  The components with $(\hat{i} \hat{j} \hat{k})$  and $(\hat{i}JK)$ automatically give:
\beq R_{L\hat{i} \hat{j} \hat{k}} = R_{L\hat{i}JK} = 0. \nonumber \eeq
\noindent The components with $(IJ\hat{k})$ are, 
\beq A_{N}^{~L} [R_{LIJ\hat{k}} - \delta_{LI} {\sf a}_{J\hat{k}} - \frac12 ({\sf s}_{I\hat{k}} + {\sf a}_{I \hat{k}}) \delta_{LJ}] = 0, \nonumber \eeq
\noindent since ${\sf s}_{\hat{k}I} + {\sf a}_{\hat{k}I} =0$ and ${\sf s}_{\hat{k}I}$ is symmetric, then ${\sf s}_{I\hat{k}} = {\sf s}_{\hat{k}I} = - {\sf a}_{\hat{k}I}$ and so:
\beq A_{N}^{~L}R_{LIJ\hat{k}} = A_{N}^{~L}[ \delta_{LI} {\sf a}_{J\hat{k}} + {\sf a}_{I\hat{k}} \delta_{LJ}]. \eeq
\noindent Symmetrizing the last two indices so that the left hand side vanishes we find that,
\beq A_{N}^{~L}R_{LI(J\hat{k})} = \delta_{LJ} {\sf a}_{I\hat{k}} = 0 \nonumber \eeq
\noindent and so the components ${\sf a}_{I\hat{k}} $ must vanish. 

In a similar manner, we find that the components with $(I\hat{j} \hat{k})$ it follows that
\beq  A_{N}^{~L} [R_{LI\hat{j} \hat{k}} - \delta_{LI} {\sf a}_{\hat{j} \hat{k}}] = 0. \nonumber \eeq
\noindent Since $R_{LI\hat{j}\hat{k}}$ is anti-symmetric with respect to the first two indices, ${\sf a}_{\hat{j} \hat{k}}$ must vanish.  This allows for a simplification in the $(\hat{i}J\hat{k})$ components,
\beq  A_{N}^{~L} [ R_{L\hat{i}J\hat{k}} -( {\sf s}_{\hat{i} \hat{k}} + {\sf a}_{\hat{i}\hat{k}}) \delta_{LJ}] = 0.  \nonumber \eeq 
\noindent Imposing the vanishing of ${\sf a}_{\hat{j} \hat{k}}$ this identity becomes
\beq A_{N}^{~L} [ R_{L\hat{i}J\hat{k}} - {\sf s}_{\hat{i} \hat{k}} \delta_{LJ}] = 0.  \nonumber \eeq 
Finally, the components with $ijk = IJK$ are of the form,
\beq A_{N}^{~L} (R_{LIJK} + 2\sigma \delta_{I[J} \delta_{|L| K]}) = 0.  \nonumber \eeq
\end{proof}

\begin{prop}  \label{prop:HomoCons}
For any Kundt Triple, with $0<dim {\cal I}_p<n$, the transverse space can be decomposed into a locally homogeneous semi-direct product, with the coordinates $(x^C,~x^{\hat{c}})$ lying in the range $C \in [1, m]$ and $\hat{c} \in [m+1, n]$:  
\beq {\bh}= h_{AB}(x^C) \d x^A \d x^B + h_{\hat{a} \hat{b}}(x^A, x^{ \hat{c}}) \d x^{\hat{a}} \d x^{\hat{b}}, \nonumber \eeq
\noindent where $h_{AB}$ is a maximally symmetric space,  and $h_{\hat{a} \hat{b}}$ is a locally homogeneous space with the following conditions on the curvature tensor:
\beq \begin{aligned} & R_{I\hat{j}K \hat{l}} = {\sf s}_{\hat{j} \hat{l}} \delta_{IK}, \\ 
& R_{IJ\hat{k} \hat{l}} = R_{I\hat{j} \hat{k} \hat{l}} = 0. \end{aligned} \label{cons:sdlh}\eeq

\end{prop}

\begin{proof}

Relative to the orthonormal coframe $\{ \tilde{m}^i\} = \{ v^1, \cdots v^m, m^{m+1},\cdots, m^n\}$, the non-zero Riemann tensor components are:  
\beq R_{IJKL}, R_{I \hat{j} K \hat{l}},\text{ and }  R_{\hat{i} \hat{j} \hat{k} \hat{l} }, \nonumber \eeq 

\noindent where 
\beq R_{IJKL} = \sigma ( \delta_{IK} \delta_{JL} - \delta_{IL} \delta_{JK}).  \label{MSspace} \eeq

\noindent As $\sigma$ is constant, $R_{IJKL; n_1 \cdots n_k} = 0,~k>0$, while $R_{I \hat{j} K \hat{l}; n_1 \cdots n_k}$ and  $R_{  \hat{i} \hat{j} \hat{j} \hat{k} ; n_1 \cdots n_k}$ are non-zero. Using the Cartan-Karlhede algorithm \cite{Kramer, Olver}, this implies that the metric is related by a coordinate transformation to some metric of the form:
\beq h_{ab} = h_{AB}(x^C) \d x^A \d x^B + h_{\hat{a} \hat{b}} (x^C, x^{\hat{c}}) \d x^{\hat{a}} \d x^{\hat{b}}, \nonumber \eeq
\noindent where $h_{AB}$ is the metric for a maximally symmetric space.

Choosing a frame such that: 
\beq m^{~I}_{A} m^{~J}_{B} \delta_{IJ} = h_{AB},~~ m^{~\hat{i}}_{\hat{a}} m^{~\hat{j}}_{\hat{b}} \delta_{\hat{i} \hat{j}} = h_{\hat{a} \hat{b}}, \nonumber \eeq
\noindent with $h_{AB}$ a maximally symmetric space. The connection coefficients $\Gamma_{ijk}$ ($\Gamma_{ijk} = - \Gamma_{jik}$) are given by the formulae:
\beq \Gamma_{ikj} = \frac12 ( D_{ijk} - D_{jki} + D_{kij}),~~ D_{ijk} = -m_{i [a;b]} m_j^{~a} m_k^{~b}. \nonumber \eeq

\noindent Relative to this metric, the non-zero connection coefficients are $\Gamma_{IJK}, \Gamma_{\hat{i} \hat{j} \hat{k}}, \Gamma_{I\hat{k} \hat{j}}$ and $\Gamma_{\hat{j} \hat{k} I}$, with the last two defined as 
\beq \begin{aligned}  \Gamma_{I\hat{k} \hat{j} } = \frac12 ( D_{\hat{j} I \hat{k}} + D_{\hat{k} I \hat{i}}), ~~ \Gamma_{\hat{j}  \hat{k} I } = \frac12 ( D_{\hat{j} I \hat{k}} - D_{\hat{k} I \hat{i}}), \end{aligned} \label{CrossCF} \eeq 

\noindent where  $D_{\hat{j} I \hat{k}} = -m_{\hat{j} \hat{a} , B} m_I^{~B} m_{\hat{k}}^{\hat{a}}$.

The non-zero Riemann tensor components are

\beq R_{IJKL},~ R_{IJ \hat{k} \hat{l}},~ R_{I \hat{j} K \hat{l}}, R_{I \hat{i} \hat{j} \hat{k}} ~ \text{ and }  R_{\hat{i} \hat{j} \hat{j} \hat{k} } \nonumber \eeq 

\noindent where $R_{IJKL}$ is of the form \eqref{MSspace}. Assuming this is a l.h. space, we must impose that $R_{IJ \hat{k} \hat{l}}$ and $R_{I \hat{i} \hat{j} \hat{k}}$ vanish and $  R_{I \hat{j} K \hat{l} } = {\sf s}_{\hat{j} \hat{l}} \delta_{IK}$.

\end{proof}

\noindent Despite the strict conditions on the l.h. transverse space, there are many possibilities for the semi-direct product of the Riemannian metrics. Some specific forms can be found by imposing conditions on the isotropy group using theorem 1.8 in \cite{DMZ}: 

\begin{thm} 
Let $G$ be a semi-simple group with finite center and no local $SL_2(\mathbb{R})$-factor, acting isometrically, faithfully, and nonproperly on a Lorentz manifold $\tilde{M}$. Then
\begin{enumerate}
\item $G$ has a local factor $G_1$ isomorphic to $O(1,n)$ with $n\geq 2$ or $O(2,n)$ with $n\geq 3$;
\item There exists a Lorentz manifold $S$, isometric, up to finite cover, to $dS_n$ or $AdS_{n+1}$, depending whether $G_1$ is isomorphic to $O(1,n)$ or $O(2,n)$, and an open subset of $\tilde{M}$ in which each $G_1$-orbit is homothetic to $S$; 
\item Any such orbit as above has a $G_1$-invariant neighbourhood isometric to a warped product $L \times_{w} S$, for $L$ a Riemannian manifold. 
\end{enumerate}
\end{thm} 

\noindent Applying this result to the l.h. Kundt triples we have the following result:
\begin{lem}
For any locally homogeneous Kundt$^\infty$ triple with a subgroup of the isotropy group isomorphic to $O(1,\tilde{n})$ with $\tilde{n} \geq 2$ or $O(2,\tilde{n})$ with $\tilde{n}\geq 3$ acting nonproperly on the manifold, it is isometric (up to finite cover) to a direct product manifold, $L \times dS_{\tilde{n}}$ or $L \times AdS_{\tilde{n}+1} $ respectively, where $L$ is a l.h. Riemannian manifold and $\dim \mathcal{I}_p = \tilde{n}$.

\end{lem}

\begin{proof}
For any warped product $L \times_w dS_{\tilde{n}}$ or $L \times_w AdS_{\tilde{n}+1} $, the transverse metric is of the form:
\beq e^{f} {\bg}' + {\bh}, \nonumber \eeq
\noindent where ${\bg}'$ is the metric for  $ dS_{\tilde{n}}$ or $AdS_{\tilde{n}+1}$ given in subsection \ref{SubSec: MaxSym}, ${\bh}$ is a Riemannian metric on $L$ and $f \in \Omega^0(L)$. In order to be CSI, we have shown that ${\bh}$ must be a l.h. metric. To determine the form of $f$, we employ a change of coordinates $e^{f} v' = v$ to bring the metric into Kundt form with 
\beq H = e^{-f} H',~ W_A = W_A' \text{ and }W_a = - f_{,a}, \nonumber  \eeq

\noindent where the indices $A$ range over the coordinate indices for ${\bg}'$ and $a$ range over the coordinate indices for ${\bh}$.  Since $H'$ is of the form:
\beq H' = \frac{\epsilon}{2} + \frac18 W_B' W^{B'},~\epsilon = -1,0,~1. \label{Hprime} \eeq

\noindent To be $CSI_0$ the function $H$ must satisfy 
\beq H =  e^{-f} H' = \frac{\sigma}{2} + \frac18 (W_B' W^{B'} + f_{,b} f^{,b}) \nonumber  \eeq

\noindent for some constant $\sigma$.  Differentiating with respect to $x^A$,
\beq H_{,A} = e^{-f}  H'_{,A} = \frac18 (W_B' W^{B'})_{,A}, \nonumber  \eeq

\noindent and substituting the form of $H'$ in \eqref{Hprime}  gives the following condition for $f$:
\beq  e^{-f} (W_B' W^{B'})_{,A} - (W_B' W^{B'})_{,A} = 0, \nonumber \eeq

\noindent this can only occur if $e^{-f} = 1$. Since $W_a = f_{,a} = 0$, and ${\bf g}'$ is maximally symmetric, it follows that $dim \mathcal{I}_p = \tilde{n}$. 

\end{proof}

When the isotropy group is not isomorphic to $SO(\tilde{n})$ or $SO(1,\tilde{n})$ it is difficult to determine additional conditions for the l.h. semi-direct product. However, in the special case of warped products we can say something definite.

\begin{thm}
For any l.h. Kundt triple with $\dim I^1|_p = m$ and $0< \dim \mathcal{I}_p < n$. If the transverse space is a locally homogeneous warped product:  
\beq {\bh}= h_{AB}(x^C) dx^A dx^B + e^{2f(x^A)} h_{\hat{a} \hat{b}}(x^{ \hat{c}}) dx^{\hat{a}} dx^{\hat{b}}, \nonumber \eeq
\noindent where the coordinates $(x^C,~x^{\hat{c}})$ lie in the range $C \in [1, m]$ and $\hat{c} \in [m+1, n]$,  $h_{AB}$ is a maximally symmetric space,  and $h_{\hat{a} \hat{b}}$ is a locally homogeneous space, then the function $f(x^A)$ must satisfy:

\beq - [ \nabla_{m_K} m_I (f) + m_K(f) m_I(f)] = S \delta_{KI},
 \nonumber \eeq

\noindent where $S$ is a constant arising from $s_{\hat{i} \hat{j}} = S \delta_{\hat{i} \hat{j}}$. 

\end{thm} 

\begin{proof}
From the proof of \eqref{prop:HomoCons}, the connection coefficients with mixed indices are

\beq \Gamma_{I \hat{k} \hat{j}} = -m_J^{~B} f_{,B} \delta_{\hat{j} \hat{k}},~~\Gamma_{\hat{j} \hat{k} I} = 0. \nonumber \eeq

\noindent Computing the curvature tensor gives one set of non-zero components with mixed-indices

\beq R_{I\hat{j}K \hat{l}} &=& m_K(\Gamma_{I\hat{j} \hat{l} }) - \Gamma^L_{~\hat{j} \hat{l}} \Gamma_{LIK} - \Gamma^{\hat{i}}_{~K \hat{l}} \Gamma_{I \hat{j} \hat{i} } \\ 
&=& [m_K (m_I(f)) - m_L(f) \Gamma^L_{~IK} - m_K(f) m_I(f) ] \delta_{\hat{j} \hat{l}}. \nonumber \eeq

\noindent Requiring that $R_{I\hat{j}K \hat{l}} = {\sf s}_{\hat{j} \hat{l}} \delta_{IK}$ and simplifying this expression completes the proof. 

\end{proof}

\noindent For any such warped product, the isotropy group will, at least, consist of the isotropy group of the maximally symmetric Riemannian space and the boosts.


Regardless of the form of the l.h. semi-direct product of the transverse space, we can state the following result: 

\begin{thm} \label{thm:inter}
Any CSI Kundt$^\infty$ metric  with $\dim I^1|_p = m$ and $0 < m \leq \dim \mathcal{I}_p < n$ has a locally homogeneous transverse metric arising from the semi-direct product with the metric, ${\bh}$, of a maximally symmetric space ($\mathbb{E}^m, S^m$ or $\mathbb{H}^m$ depending on the sign of $\sigma$) and the metric, $\bh'$, of a locally homogeneous space:
\small
\beq g &=& 2 \d u \left[ \d v + v^2\left[\frac{\sigma}{2} + \frac18 ( W^A W_A + {W'}^{\hat{a}} W'_{\hat{a}})\right]\d u + v (W_A dx^A + W'_{\hat{a}} \d x^{\hat{a}})\right] \nonumber\\  && + h_{AB}(x^C) \d x^A \d x^B + h'_{\hat{a} \hat{b}} (x^A, x^{\hat{c}})  \d x^{\hat{a}} \d x^{\hat{b}}.\eeq
\normalsize
\noindent  The one-form ${\bW} = \d \phi$ is defined on the maximally symmetric space through the differential equations: 
\beq f_{;AB} = - \sigma h_{AB} f,~~ f = \
\exp(-\phi/2), \eeq 
\noindent which has solutions for each value of $\sigma$ given in subsection \ref{SubSec: MaxSym}.  The remaining one-form ${\bW'}$ is a left-invariant one-form of the entire transverse space.
\end{thm}

\begin{proof}

From Proposition \ref{prop:HomoCons} it follows the transverse space must be a semi-direct product of a maximally symmetric space with a l.h. space. Then Lemma \ref{lem:Maxsym} and Theorem \ref{thm:cone} give the corresponding form for the part of ${\bW}$ projected to the $v^I$ basis.

The remaining one-form, ${\bW'}$, relative to the $m^{\hat{i}}$ basis must be left-invariant in order to preserve the condition that $\dim I^1|_p = m$. To see why consider the expression for $\tilde{R}_{121i;2} = \alpha_i$ where $\bar{W} = \bW + {\bW'}$ 
\beq \sigma \bar{W}_i - 2({\sf a}_{ij} + {\sf s}_{ij}) \bar{W}^j = {\sf  \alpha}_i. \nonumber \eeq

\noindent As we have a basis of eigenvectors $v_I^{~i} = \pounds_{\xi_I} \bar{{W}}^i $ satisfying 
\beq \sigma v_{I~i} - 2({\sf a}_{ij} + {\sf s}_{ij}) v_I^{~i}, \nonumber \eeq
\noindent we can choose an orthonormal basis adapted to the eigenvectors, implying that ${\sf \alpha}_I = 0$ since
\beq {\sf a}_{IJ} + {\sf s}_{IJ} = \frac{\sigma}{2} \delta_{IJ}. \nonumber \eeq 

For the remaining components, ${\sf \alpha}_{\hat{i}}$ we have the following expression
\beq  \bar{W}_{\hat{i}} -  2{\sf s}_{\hat{i}\hat{j}} \bar{W}^{\hat{j}} = \alpha_{\hat{i}}, \nonumber \eeq
\noindent where $a_{\hat{i}\hat{j}}$ vanishes due to the basis as chosen in the proof of lemma \ref{lem:Curv}. Taking the Lie derivative of ${\sf \alpha}_{\hat{i}} $ with respect to $\xi_{\hat{i}}$ gives

\beq \sigma \pounds_{\xi_{\hat{k}}} \bar{W}_{\hat{i}} - 2 {\sf s}_{\hat{i} \hat{j}} \pounds_{\xi_{\hat{k}}} \bar{W}^{\hat{j}} = 0. \nonumber \eeq

\noindent As we have assumed that $v_I$ forms an $m$-dimensional basis of the eigenvectors for the eigenvalue $\sigma$ of ${\sf a}_{ij} + {\sf s}_{ij}$ in $I^1|_{p}$ it follows that $\pounds_{\xi_{\hat{k}}} \bar{W}_{\hat{i}}$ must vanish, and $\bar{W}_{\hat{i}} = W'_{\hat{i}}$ is left-invariant.

\end{proof}

Examples of these intermediately ranked spaces can be found among the Lorentzian solvmanifolds admitting extra symmetries \cite{Solve}. For example, 
\[ 2\d u(\d v+2pv\d x)+\d x^2+e^{-2px}\sum_{A=1}^m(\d y^A )^2+\sum_{i=1}^{n-m-1}e^{-2q_ix}(\d z^i )^2,\] 
which has $\dim\mathcal{I}_p=m$. Other examples can also be found by considering more complicated solvmanifolds, see \cite{LauretA, LauretB, Solve} for how they can be constructed\footnote{In these papers Einstein solvmanifolds are emphasized but the general constuction is outlined.}.

\section{Conclusion} \label{sec:Conclusion}

For any $\mathcal{I}$-degenerate CSI spacetime, we can employ the diffeomorphism $\phi_t$ generated by a boost at a point $\tilde{p} \in \tilde{M}$ and take the limit as $t \to \infty$ to generate a new l.h. CSI spacetime which is of alignment type {\bf D} to all orders, called a l.h. Kundt${^\infty}$ triple. Assuming that all $\mathcal{I}$-degenerate spacetimes belong to the degenerate Kundt class, any other CSI spacetime (which is not l.h.) can be constructed by adding additional metric functions to the form of the Kundt${^\infty}$ triple metric with the same SPIs. 

Motivated by this fact, we have classified the $(n+2)$-dimensional l.h. Kundt${^\infty}$ triples by determining the permitted forms for $(\sigma, \bW, \bh)$ using the CSI$_0$ and CSI$_1$ conditions \cite{CSI4b}. Taking the Lie derivative of the CSI$_1$ condition \eqref{eq:alpha} with respect to the Killing vector fields yields an eigenvector problem for the eigenvalue $\sigma$ and a corresponding eigenbasis, $\mathcal{I}_p$ in the tangent space  at each point given in \eqref{Ipset}. 

When $dim~\mathcal{I}_p = n$ the l.h. Kundt${^\infty}$ triples correspond to the maximally symmetric spacetimes, and $\bW$ is the gradient of a function $f$ satisfying a differential equation \eqref{eq:s}. At the other extreme, if $dim~\mathcal{I}_p = 0$ the l.h. Kundt triples correspond to the well-known case where $\bW$ is a left-invariant one-form on the l.h. transverse space $\bh$. In the intermediate case $0< dim~\mathcal{I}_p<n$, the general form of the l.h. Kundt${^\infty}$ triples is described by theorem \ref{thm:inter}.  

In future work, we will examine the set of l.h. Kundt${^\infty}$ triples which are Einstein, and investigate the existence of universal spacetimes within this restricted class of CSI spacetimes for composite dimension \cite{hervik2015}. Noting that the l.h. Kundt${^\infty}$ triples belong to the type ${\bf D}^k$ spacetimes and hence the curvature tensor and its covariant derivatives to all orders are of alignment type {\bf D} \cite{CHP2010}, we will determine the explicit form of the type ${\bf D}^k$ metrics as a subclass of the degenerate Kundt spacetimes \cite{typedk}. We note that the limiting process used to generate the l.h. Kundt${^\infty}$ triples is not applicable to other CSI pseudo-Riemannian metrics. However, it is of interest to identify the subset of pseudo-Riemannian metrics whose limit are l.h. metrics, this will be explored in the case of 4D neutral signature metrics.

\section*{Acknowledgements}

This work was supported through the Research Council of Norway, Toppforsk grant no. 250367: Pseudo-
Riemannian Geometry and Polynomial Curvature Invariants: Classification, Characterisation and Applications. 

\bibliographystyle{unsrt-phys}
\bibliography{KtripleReferences}

\begin{thebibliography}{10}

\bibitem{TP}
F.~Pr{\"u}fer, F.~Tricerri, and L.~Vanhecke.
\newblock {Curvature invariants, differential operators and local homogeneity}.
\newblock {\em Transactions of the American Mathematical Society},
  348(11):4643--4652, 1996.

\bibitem{Hervik:2010gz}
S.~Hervik and A.~Coley.
\newblock {Pseudo-Riemannian VSI spaces}.
\newblock {\em Class. Quant. Grav.}, 28:015008, 2011.
\newblock \href{http://arxiv.org/abs/gr-qc/1008.2838}{{arXiv:1008.2838
  [gr-qc]}}.

\bibitem{Hervik:2012zz}
S.~Hervik.
\newblock {Pseudo-Riemannian VSI spaces II}.
\newblock {\em Class. Quant. Grav.}, 29:095011, 2012.
\newblock \href{http://arxiv.org/abs/math-ph/1504.01616}{{arXiv:1504.01616
  [math-ph]}}.

\bibitem{HHY}
S.~Hervik, A.~Haarr, and K.~Yamamoto.
\newblock {I-degenerate pseudo-Riemannian metrics}.
\newblock {\em Journal of Geometry and Physics}, 98:384--399, 2015.
\newblock \href{http://arxiv.org/abs/1410.4347}{{arXiv:1410.4347 [gr-qc]}}.

\bibitem{CSI4a}
A.~{Coley}, S.~{Hervik}, and N.~{Pelavas}.
\newblock {Spacetimes characterized by their scalar curvature invariants}.
\newblock {\em Classical and Quantum Gravity}, 26(2):025013, 2009.
\newblock \href{http://arxiv.org/abs/0901.0791}{{arXiv:0901.0791 [gr-qc]}}.

\bibitem{CSI4b}
A.~Coley, S.~Hervik, and N.~Pelavas.
\newblock Lorentzian spacetimes with constant curvature invariants in four
  dimensions.
\newblock {\em Classical and Quantum Gravity}, 26(12):125011, 2009.
\newblock \href{http://arxiv.org/abs/0904.4877}{{arXiv:0904.4877 [gr-qc]}}.

\bibitem{CSI4c}
A.~{Coley}, S.~{Hervik}, and N.~{Pelavas}.
\newblock {Lorentzian spacetimes with constant curvature invariants in three
  dimensions}.
\newblock {\em Classical and Quantum Gravity}, 25(2):025008, 2008.
\newblock \href{http://arxiv.org/abs/0710.3903}{{arXiv:0710.3903 [gr-qc]}}.

\bibitem{CSI4d}
A.~Coley, S.~Hervik, and N.~Pelavas.
\newblock On spacetimes with constant scalar invariants.
\newblock {\em Classical and Quantum Gravity}, 23(9):3053, 2006.
\newblock \href{http://arxiv.org/abs/gr-qc/0509113}{{arXiv:0509113 [gr-qc]}}.

\bibitem{Kramer}
H.~Stephani, D~Kramer, M.~MacCallum, C.~Hoenselaers, and E.~Herlt.
\newblock {\em Exact solutions of Einstein's field equations}.
\newblock Cambridge University Press, 2009.

\bibitem{McNutt:2012}
D.~McNutt, R.~Milson, and A.~Coley.
\newblock {Vacuum Kundt Waves}.
\newblock {\em Class. Quant. Grav.}, 30:055010, 2013.
\newblock \href{http://arxiv.org/abs/gr-qc/1208.5027}{{arXiv:1208.5027
  [gr-qc]}}.

\bibitem{Podolsky:2007}
J.~Podolsky and D.~Kofron.
\newblock {Chaotic motion in Kundt spacetimes}.
\newblock {\em Class. Quant. Grav.}, 24:3413--3424, 2007.
\newblock \href{http://arxiv.org/abs/gr-qc/0705.3098}{{arXiv:0705.3098
  [gr-qc]}}.

\bibitem{Sakalli:2007}
I.~Sakalli and M.~Halilsoy.
\newblock {Chaos in Kundt type III Spacetimes}.
\newblock {\em Chin. Phys. Lett.}, 28:070402, 2011.
\newblock \href{http://arxiv.org/abs/gr-qc/0706.2511}{{arXiv:0706.2511
  [gr-qc]}}.

\bibitem{siklos1981}
S.~T.~C. Siklos.
\newblock {Some Einstein spaces and their global properties}.
\newblock {\em Journal of Physics A: Mathematical and General}, 14(2):395,
  1981.

\bibitem{Podolsky:1997}
J.~Podolsky.
\newblock {Interpretation of the Siklos solutions as exact gravitational waves
  in the anti-de Sitter universe}.
\newblock {\em Class. Quant. Grav.}, 15:719--733, 1998.
\newblock \href{http://arxiv.org/abs/gr-qc/gr-qc/9801052}{{arXiv:gr-qc/9801052
  [gr-qc]}}.

\bibitem{Calvaruso:2019}
Giovanni Calvaruso.
\newblock Siklos spacetimes as homogeneous ricci solitons.
\newblock {\em Classical and Quantum Gravity}, 36(9):095011, 2019.

\bibitem{Baleanu:2001}
D.~Baleanu and S.~Baskal.
\newblock {Dual metrics and nongeneric supersymmetries for a class of Siklos
  space-times}.
\newblock {\em Int. J. Mod. Phys.}, A17:3737--3748, 2002.
\newblock \href{http://arxiv.org/abs/gr-qc/0108029}{{arXiv:gr-qc/0108029
  [gr-qc]}}.

\bibitem{CFH0}
A.~A. Coley, A.~Fuster, S.~Hervik, and N.~Pelavas.
\newblock {Vanishing Scalar Invariant Spacetimes in Supergravity}.
\newblock {\em JHEP}, 05:032, 2007.
\newblock
  \href{http://arxiv.org/abs/hep-th/hep-th/0703256}{{arXiv:hep-th/0703256
  [hep-th]}}.

\bibitem{CFH}
A.~{Coley}, A.~{Fuster}, and S.~{Hervik}.
\newblock {Supergravity Solutions with Constant Scalar Invariants}.
\newblock {\em International Journal of Modern Physics A}, 24:1119--1133, 2009.
\newblock \href{http://arxiv.org/abs/0707.0957}{{arXiv:0707.0957 [gr-qc]}}.

\bibitem{horowitz1990spacetime}
G.~T. Horowitz and A.~R. Steif.
\newblock {Spacetime singularities in string theory}.
\newblock {\em Physical Review Letters}, 64(3):260, 1990.

\bibitem{coley2008metrics}
A.~A. Coley, G.~W. Gibbons, S.~Hervik, and C.~N. Pope.
\newblock {Metrics with vanishing quantum corrections}.
\newblock {\em Classical and Quantum Gravity}, 25(14):145017, 2008.

\bibitem{coleyhervik2011}
A.~A. Coley and S.~Hervik.
\newblock Universality and constant scalar curvature invariants.
\newblock {\em ISRN Geometry}, 2011, 2011.
\newblock \href{http://arxiv.org/abs/1105.2356}{{arXiv:1105.2356 [gr-qc]}}.

\bibitem{Hervik:2013}
S.~Hervik, V.~Pravda, and A.~Pravdova.
\newblock {Type III and N universal spacetimes}.
\newblock {\em Class. Quant. Grav.}, 31(21):215005, 2014.
\newblock \href{http://arxiv.org/abs/gr-qc/1311.0234}{{arXiv:1311.0234
  [gr-qc]}}.

\bibitem{Hervik:2015}
S.~Hervik, V.~Pravda, and A.~Pravdová.
\newblock {Type N universal spacetime}.
\newblock {\em J. Phys. Conf. Ser.}, 600(1):012065, 2015.

\bibitem{hervik2015}
S.~Hervik, T.~M{\'a}lek, V.~Pravda, and A.~Pravdov{\'a}.
\newblock {Type II universal spacetimes}.
\newblock {\em Classical and Quantum Gravity}, 32(24):245012, 2015.
\newblock \href{http://arxiv.org/abs/1503.08448}{{arXiv:1503.08448 [gr-qc]}}.

\bibitem{Hervik:2017a}
S.~Hervik, T.~Málek, V.~Pravda, and A.~Pravdová.
\newblock {On type II universal spacetimes}.
\newblock In {\em {Proceedings, 14th Marcel Grossmann Meeting on Recent
  Developments in Theoretical and Experimental General Relativity,
  Astrophysics, and Relativistic Field Theories (MG14) (In 4 Volumes): Rome,
  Italy, July 12-18, 2015}}, volume~3, pages 2542--2547, 2017.

\bibitem{Hervik:2017b}
S~Hervik, V.~Pravda, and A.~Pravdová.
\newblock {Universal spacetimes in four dimensions}.
\newblock {\em JHEP}, 10:028, 2017.
\newblock \href{http://arxiv.org/abs/gr-qc/1707.00264}{{arXiv:1707.00264
  [gr-qc]}}.

\bibitem{Hervik:2018}
S.~Hervik, M.~Ortaggio, and V.~Pravda.
\newblock {Universal electromagnetic fields}.
\newblock {\em Class. Quant. Grav.}, 35(17):175017, 2018.
\newblock \href{http://arxiv.org/abs/gr-qc/1806.05835}{{arXiv:1806.05835
  [gr-qc]}}.

\bibitem{CHP2010}
A.~{Coley}, S.~{Hervik}, and N.~{Pelavas}.
\newblock {Lorentzian manifolds and scalar curvature invariants}.
\newblock {\em Classical and Quantum Gravity}, 27(10):102001, 2010.
\newblock \href{http://arxiv.org/abs/1003.2373}{{arXiv:1003.2373 [gr-qc]}}.

\bibitem{classa}
A.~{Coley}, R.~{Milson}, V.~{Pravda}, and A.~{Pravdov{\'a}}.
\newblock {Classification of the Weyl tensor in higher dimensions}.
\newblock {\em Classical and Quantum Gravity}, 21:L35--L41, 2004.~.
\newblock \href{http://arxiv.org/abs/gr-qc/0401008}{{arXiv:0401008 [gr-qc]}}.

\bibitem{classb}
R.~{Milson}, A.~{Coley}, V.~{Pravda}, and A.~{Pravdova}.
\newblock {Alignment and algebraically special tensors in Lorentzian geometry}.
\newblock {\em International Journal of Geometric Methods in Modern Physics},
  2(01):41--61, 2005.~.
\newblock \href{http://arxiv.org/abs/gr-qc/0401010}{{arXiv:0401010 [gr-qc]}}.

\bibitem{classc}
A.~{Coley}.
\newblock {Classification of the Weyl Tensor in Higher Dimensions and
  Applications}.
\newblock {\em Classical and Quantum Gravity}, 25(3):033001, 2008.
\newblock \href{http://arxiv.org/abs/0710.1598}{{arXiv:0710.1598 [gr-qc]}}.

\bibitem{OrtaggioPravdaPravdova:2013}
M.~Ortaggio, V.~Pravda, and A.~Pravdov{\'a}.
\newblock Algebraic classification of higher dimensional spacetimes based on
  null alignment.
\newblock {\em Classical and Quantum Gravity}, 30(1):013001, 2012.
\newblock \href{http://arxiv.org/abs/1211.7289}{{arXiv:1211.7289 [gr-qc]}}.

\bibitem{Hervik2011}
S.~Hervik.
\newblock {A spacetime not characterized by its invariants is of aligned type
  II}.
\newblock {\em Classical and Quantum Gravity}, 28(21):215009, 2011.
\newblock \href{http://arxiv.org/abs/1109.2551}{{arXiv:1109.2551 [gr-qc]}}.

\bibitem{CHPP2009}
A.~{Coley}, S.~{Hervik}, G.~{Papadopoulos}, and N.~{Pelavas}.
\newblock {Kundt spacetimes}.
\newblock {\em Classical and Quantum Gravity}, 26(10):105016, 2009.
\newblock \href{http://arxiv.org/abs/0901.0394}{{arXiv:0901.0394 [gr-qc]}}.

\bibitem{Higher}
A.~{Coley}, R.~{Milson}, V.~{Pravda}, and A.~{Pravdov{\'a}}.
\newblock {Vanishing scalar invariant spacetimes in higher dimensions}.
\newblock {\em Classical and Quantum Gravity}, 21:5519--5542, 2004.
\newblock \href{http://arxiv.org/abs/gr-qc/0410070}{{arXiv:0410070 [gr-qc]}}.

\bibitem{Podolsky:2014}
J.~Podolsky and R.~Svarc.
\newblock {Algebraic structure of Robinson-Trautman and Kundt geometries in
  arbitrary dimension}.
\newblock {\em Class. Quant. Grav.}, 32:015001, 2015.
\newblock \href{http://arxiv.org/abs/gr-qc/1406.3232}{{arXiv:1406.3232
  [gr-qc]}}.

\bibitem{MPS2018}
M.~Mars, T.~T. Paetz, and J.~M.~M. Senovilla.
\newblock Multiple killing horizons and near horizon geometries.
\newblock {\em Classical and Quantum Gravity}, 35(15):155015, 2018.
\newblock \href{http://arxiv.org/abs/1807.02679}{{arXiv:1807.02679 [gr-qc]}}.

\bibitem{Lewandowski2018}
J.~Lewandowski and A.~Szereszewski.
\newblock Spacetimes foliated by non-expanding null surfaces in the presence of
  a cosmological constant.
\newblock 2018.
\newblock \href{http://arxiv.org/abs/1809.07666}{{arXiv:1809.07666 [gr-qc]}}.

\bibitem{SH2018}
S.~Hervik.
\newblock A new class of infinitesimal group actions on pseudo-riemannian
  manifolds.
\newblock {\em preprint}, 2018.
\newblock \href{http://arxiv.org/abs/1805.09402}{{arXiv:1805.09402 [gr-qc]}}.

\bibitem{MMTA2018}
D.~D. McNutt and M.~T. Aadne.
\newblock {$\mathcal{I}$-Preserving Diffeomorphisms of Lorentzian Manifolds}.
\newblock {\em J. Math. Phys.}, 60:032501, 2019.
\newblock \href{http://arxiv.org/abs/1901.04728}{{arXiv:1901.04728 [gr-qc]}}.

\bibitem{MTA2018}
M.~T. Aadne.
\newblock {Nil-Killing vector fields and Kundt spacetimes}.
\newblock {\em preprint}, 2018.

\bibitem{deMedeiros}
P.~de~Medeiros.
\newblock {Private Communication}.
\newblock 2017.

\bibitem{Olver}
P.~J. Olver.
\newblock {\em {Equivalence, invariants and symmetry}}.
\newblock Cambridge University Press, 1995.

\bibitem{DMZ}
M.~Deffaf, K.~Melnick, and A.~Zeghib.
\newblock {Actions of noncompact semisimple groups on Lorentz manifolds}.
\newblock {\em Geometric and Functional Analysis}, 18(2):463--488, 2008.

\bibitem{Solve}
S.~Hervik.
\newblock {Solvegeometry gravitational waves}.
\newblock {\em Classical and Quantum Gravity}, 21(17):4273, 2004.~.
\newblock \href{http://arxiv.org/abs/gr-qc/0403065}{{arXiv:0403065 [gr-qc]}}.

\bibitem{LauretA}
J.~Lauret.
\newblock {Einstein solvmanifolds are standard}.
\newblock {\em Annals of mathematics}, 172:1859--1877, 2010.~.
\newblock \href{http://arxiv.org/abs/math/0703472}{{arXiv:0703472 [math]}}.

\bibitem{LauretB}
J.~Lauret.
\newblock {Standard Einstein solvmanifolds as critical points}.
\newblock {\em The Quarterly Journal of Mathematics}, 52(4):463--470, 2001.

\bibitem{typedk}
D~{McNutt}, A.~{Coley}, L~{Wylleman}, and S.~{Hervik}.
\newblock {Locally Boost Isotropic ${\bf D}^k$ Spacetimes}.
\newblock 2019.
\newblock \href{http://arxiv.org/abs/1907.08957}{{arXiv:1907.08957 [gr-qc]}}.

\end{thebibliography}

\end{document}